\documentclass[12pt]{svproc}
\pdfoutput=1
\usepackage[utf8]{inputenc}
\usepackage[]{fontenc}
\usepackage{array}
\usepackage{amsmath}
\usepackage{amssymb}
\usepackage{xspace}
\usepackage{ifdraft}%
\usepackage{proof}
\usepackage{color}
\usepackage{hyperref}
\usepackage{tikz}
\usepackage{wrapfig}
\newcommand{\<}{\langle\xspace}
\renewcommand{\>}{\rangle\xspace}
\newcommand{\succc}{\;\mathrm{succ}\;\xspace}

\DeclareMathOperator{\inlOp}{\mathsf{in}_1}
\DeclareMathOperator{\inrOp}{\mathsf{in}_2}
\def\inl#1#2{\inlOp_{,#1}#2}
\def\inr#1#2{\inrOp_{,#1}#2}

\def\caseOf#1#2#3#4#5{\mathsf{case}\; #1\;
    \mathsf{of}\; \left[#2\right]#3,\,\left[#4\right]#5}

\def\letBeIn#1#2#3{\mathsf{let}\; #1\;\mathsf{be}\; #2\; \mathsf{in}\; #3\xspace}
\def\packTo#1#2#3#4{\mathsf{pack}\; #1,\,#2\;\mathsf{to}\;\exists#3.\,#4}

\newcommand{\bbot}{\mbox{{$\bot\hspace{-1.5ex}\bot$}}\xspace}
\def\FV#1{\mathrm{FV}(#1)}

\def\cal#1{\mathcal{#1}}
\def\tuple#1{\left\langle#1\right\rangle\xspace}
\def\defwhen{\quad\textrm{when}\quad}
\def\sqr#1#2{\vbox
 {\hrule height#2
  \hbox{\vrule width#2 height#1 \kern#1 \vrule width#2}%
  \hrule height#2}}
\def\fmp{finite model property\xspace}
\def\efmp{\emph{\fmp}\xspace}
\def\smp{small model property\xspace}
\def\esmp{\emph{\smp}\xspace}

\def\eM{\mathcal M\xspace}
\def\te{\mathfrak t}
\def\lthen{\rightarrow}
\def\figref#1{Fig.~\ref{#1}\xspace}
\def\position{\mathcal{P}}
\def\gameturn{\,\rightsquigarrow\,}
\def\gameturnMove{\,\rightsquigarrow_{\mathrm{move}}\,}
\def\vdashFO{\vdash_\mathrm{IFOL}}
\def\vdashFOx#1{\vdash_{\mathrm{IFOL},#1}}
\newcommand{\pvar}{{\cal X}_\mathrm{p}\xspace}
\begin{document}
\title{Small model property reflects in games and automata}
\author{Maciej Zielenkiewicz}
\institute{Institute of Informatics, University of Warsaw, Warsaw, Poland \email{maciekz@mimuw.edu.pl}}
\maketitle
\begin{abstract}
	Small model property is an important property that implies decidability. We  
	show that the small model size is directly related to some important resources 
	in games and automata for checking provability. 
\end{abstract}

\section{Introduction}

Dependent types is one of the popular logic-based approaches developed in the field of functional programming.
With the help of such types it is possible to more precisely capture the behaviour of programs. Intuitionistic first
order logic is the primary form of dependent types and the
Curry-Howard isomorphism strictly relates functional program synthesis and construction of proofs in intuitionistic first
order logic. 

One of the well established ways towards understanding different aspects of logic, proofs and proof search is through
correspondence with different representations, e.g. ones that are more abstract like games and tableaux or ones that are
more detailed like linear logic. One of the most fruitful ideas fulfilling the pattern is the game based approach, in
the spirit of Ehrenfeucht-Fra\"iss\'e games \cite{Fraisse,Ehrenfeucht}. Another game-based technique was
introduced for intuitionistic first order logic \cite{Urzyczyn2016}. The duality between proof-search and countermodel
search \cite{vanBenthem} has been interpreted there in terms of games and was used to make one unified game that yields
either a proof or a Kripke countermodel. 

We extend the game based approach \cite{Urzyczyn2016} to classes that have the \emph{finite model property} which, for
algorithmically well-behaved classes, implies decidability \cite[p.~240]{Borger97}. A stronger property, the \emph{small
model property}, that also gives an upper bound on the complexity of the satisfiability problem, is also studied. As it
turns out these two properties are equivalent for many interesting classes.

We show in the current work a correspondence between the limit of the model size given by the small model property and some resources
in automata and games used for the description of logic. Section~\ref{sec:preliminaries} contains preliminaries and definitions.
Section~\ref{sec:small-afrodite-s} discusses the automata and Theorem~\ref{thm:final1}
bounds the size of the set of
eigenvariables with a number dependent on the number of subformulas in the formula, the limit on the model size and the
number of variables in the initial formula. Section~\ref{sec:small-automata} covers games and Theorem~\ref{thm:final2} shows that a strategy can be constructed that uses a number of
maximal variables at most equal to the limit on the model size; the maximal variable would be understood as the one havina a
maximal, by inclusion, set of known facts.

The paper is structured as follows. Section~\ref{sec:preliminaries} contains preliminaries and definitions necessary to
understand the following sections and discusses basic facts about the small model property. 
Section~\ref{sec:small-afrodite-s} defines a quasiorder on variables capturing the notion of variable with more facts.
Using this order we show that the size of the small countermodel defined in the small model property is
also a limit on the number of \emph{maximal} variables in Afrodite strategy.
Section~\ref{sec:small-automata} shows a limit on the size of the set of eigenvariables of an Arcadian automaton,
which depends on the size of the small model, number of subformulas in the original formula and 
the number of its variables.

\section{Preliminaries}
\label{sec:preliminaries}

We work in intuitionistic first-order logic with no function symbols or constants. The logic is the same as in previous
works on games \cite{Urzyczyn2016} and automata \cite{poprzedni}. There is a set of predicates
$\mathcal P$ and every predicate $P\in\mathcal P$ has a defined \emph{arity}. First order variables are noted as 
$X$, $Y$, \dots (with possible annotations) and form an infinite set $\mathcal X_1$. 
The formulas are understood as abstract syntax trees and the possible formulas are generated with the grammar
\[ 
\tau,\sigma::= P(X,\ldots,X) \mid
                \tau\land\sigma \mid
                \tau\lor\sigma \mid
                \tau\to \sigma \mid 
                \forall X.\tau \mid
                \exists X.\tau \mid 
                \bot.
\]
We define the set $\FV{\tau}$ of free variables of a formula $\tau$ as
\begin{itemize}
\item $\FV{P(X_1,\ldots, X_n)} = \{ X_1,\ldots, X_n\}$,
\item $\FV{\tau_1\ast \tau_2} =
  \FV{\tau_1}\cup\FV{\tau_2}$ where $\ast\in\{\land,\lor,\to\},$
\item $\FV{\triangledown X.\tau} = \FV{\tau}\backslash \{X\}$ where
  $\triangledown\in\{\exists,\forall\}$,
\item $\FV{\bot} = \emptyset$.
\end{itemize}
We assume that there is an infinite set $\pvar$ of \emph{proof term variables}
usually noted as $x, y,$\ldots that can be used to form the
following terms.
$$
\begin{array}{l@{\;}l}
M,& N,P ::= x\mid 
          \<M, N\> \mid \pi_1 M\mid \pi_2 M \mid 
          \lambda x:\varphi.M \mid MN \mid 
          \lambda X M \\& MX \mid  
          \inl{\varphi\lor\psi}M \mid
          \inr{\varphi\lor\psi}M \mid 
          \caseOf{M}{x:\varphi}{N}{y:\psi}{P} \mid \\ &
          \packTo{M}{Y}{X}{\varphi} \mid
          \letBeIn{x:\varphi}{M:\exists X.\varphi}{N} \mid 
	  \bbot_\varphi M
\end{array}
$$
The free variables in terms are defined by structural recursion on the terms, i.e. 
$\FV{\lambda x:\varphi.M}=\FV \varphi \cup \FV M$.
The inference rules for the logic are shown in Figure~\ref{fig:rules}.

\begin{figure}[htb!]

  \begin{center}
$$
\infer[(var)]{\Gamma, x\!:\!\tau\vdash x:\tau}{}
$$
\vspace{-1ex}
$$
\infer[(\land I)]{\Gamma\vdash \<M_1, M_2\> :\tau_1\land \tau_2}
{\Gamma\vdash M_1:\tau_1  &
\Gamma\vdash M_2:\tau_2}
$$
$$
\infer[(\land E1)]{\Gamma\vdash \pi_1 M:\tau_1}
{\Gamma\vdash M: \tau_1\land \tau_2}
\quad
\infer[(\land E2)]{\Gamma\vdash \pi_2M:\tau_2}
{\Gamma\vdash M: \tau_1\land \tau_2}
$$
\vspace{-1ex}
$$
\infer[(\lor I1)]{\Gamma\vdash \inl{\tau_1\lor\tau_2}M :\tau_1\lor \tau_2}
{\Gamma\vdash M:\tau_1}
\qquad
\infer[(\lor I2)]{\Gamma\vdash \inr{\tau_1\lor\tau_2}M :\tau_1\lor \tau_2}
{\Gamma\vdash M:\tau_2}
$$
$$
\infer[(\lor E)]{\Gamma\vdash
  \caseOf{M}{x_1:\tau_1}{N_1}{x_2:\tau_2}{N_2}:\tau}
{\Gamma\vdash M: \tau_1\lor \tau_2 & 
  \Gamma,x_1:\tau_1\vdash N_1: \tau &
  \Gamma,x_2:\tau_2\vdash N_2: \tau}
$$
\vspace{-1ex}
$$
\infer[(\to I)]{\Gamma\vdash\lambda x:\tau_1.M:\tau_1\to \tau_2}
{\Gamma, x\!:\! \tau_1\vdash M:\tau_2}
$$
\vspace{-1ex}
$$
\infer[(\to E)]{\Gamma\vdash M_1M_2:\tau_2}
{\Gamma\vdash M_1: \tau_1\to \tau_2 & \Gamma\vdash M_2: \tau_1}
$$
\vspace{-1ex}
$$
\infer[(\forall I)^*]{\Gamma\vdash \lambda X M : \forall X.\tau}
{\Gamma\vdash M:\tau}
\qquad
\infer[(\forall E)]
{\Gamma\vdash MY:\tau[X:=Y]}
{\Gamma\vdash M:\forall X.\tau}
$$
\vspace{-1ex}
$$
\infer[(\exists I)]{\Gamma\!\vdash\! \packTo{\!M\!}{\!Y\!}{X}{\tau} : \exists
  X.\tau}%
{\Gamma\!\vdash M:\tau[X:=Y]}
$$
\vspace{-1ex}
$$
\infer[(\exists E)^*]
{\Gamma\!\vdash\! \letBeIn{x\!:\!\tau}{M_1\!:\!\exists X.\tau}{M_2}:\sigma}
{\Gamma\!\vdash M_1:\exists X.\tau & \Gamma, x\!:\!\tau\vdash M_2:\sigma}
$$
\infer[(\bot E)]
{\Gamma\vdash \bbot_\tau M:\tau}
{\Gamma\vdash M:\bot}
\end{center}

\rule{5em}{0.5pt}\\
${}^*$ Under the eigenvariable condition
$X\not\in \FV{\Gamma, \sigma}$.

  \caption{The rules of the intuitionistic first-order logic (\cite{poprzedni})}
  \label{fig:rules}
\end{figure}

\subsection{Models}
We follow the definition of Kripke model from the work of Sørensen and Urzyczyn \cite{UrzyczynSorensen}:
A Kripke model is a triple $\tuple{C,\,\leq,\,\{\mathcal A_c: c\in C\}}$ where $C\neq\emptyset$ is a set of states,
$\leq$ is a partial order on $C$ and $\mathcal A_c=\langle A_c, P_1^{\mathcal A_c}, \ldots, P_n^{\mathcal A_c}\rangle$ are 
\emph{structures} such that if $c \leq c'$ then $A_c \subseteq A_{c'}$ and for all $i$ the relation
$P_i^{\mathcal A_c} \subseteq P_i^{\mathcal A_{c'}}$ holds.
A \emph{valuation} $\rho$ maps variables to elements of $A_c$. 
The \emph{satisfaction relation} $c,\rho\models\varphi$ is defined in the usual way:
\[
	\begin{array}{lll}
		c,\rho\models P(t_1, \mathinner{{\ldotp}{\ldotp}{\ldotp}}, t_n) &\textrm{iff} & \mathcal A_c,\rho\models P( t_1, \mathinner{{\ldotp}{\ldotp}{\ldotp}}, t_n) \textrm{ classicaly},\\
		c,\rho\models\tau\lor\sigma  &\textrm{iff} & \mathcal A_c,\rho\models\tau \textrm{ or } \mathcal A_c,\rho\models\sigma,\\
		c,\rho\models\tau\land\sigma &\textrm{iff} &  \mathcal A_c,\rho\models\tau \textrm{ and } \mathcal A_c,\rho\models\sigma,\\
		c,\rho\models\tau\to\sigma   &\textrm{iff} & \textrm{for all } c'\geq c \textrm{ if } c',\rho\models\tau \textrm{ then } c',\rho\models\sigma,\\
		c,\rho\models\forall a \tau&\textrm{iff} & \textrm{for all } c'\geq c \textrm{ if } \hat a \in \mathcal A_{c'} \textrm{ then } c',\rho[\hat a/a]\models\tau,\\
		c,\rho\models\exists a \tau&\textrm{iff} & \textrm{for some } \hat a\in\mathcal A_c,\,c,\rho[\hat a/a]\models\tau\textrm.\\
	\end{array}
\]

\begin{proposition}[completeness, Theorem 8.6.7 of of \cite{UrzyczynSorensen}]
	The Kri\-pke models as defined above are complete for the intuitionistic predicate logic, i.e. $\Gamma\models\tau$ iff $\Gamma\vdash\tau$.
\end{proposition}

\subsection{The finite model property and the small model property}
We focus on classes of formulas that have \efmp. Our definitions closely follow that of Börger et al. \cite{Borger97}:
\begin{definition}[finite model property]
	A class of formulas $X$ has the \efmp when, for all formulas $\tau\in X$, if $\tau$ is satisfiable, there exists a finite model $\eM$
	such that $\eM \models \tau$.
\end{definition}
Since all classical theories can be easily expressed as intuitionistic theories by explicitly
including the law of excluded middle, so there are many interesting classes that have finite model
property.

Although the finite model property in the book by B\"orger et al.~\cite{Borger97} is strongly attached to decidability of a particular fragment of logic, this
is not a property that implies this computational feature. The following property is what actually takes
place in the fragment considered in the book.
\begin{definition}[small model property]
	A class $X$ has the \esmp when there exists a computable function $s_X$  such that for all formulas $\tau\in X$, if $\tau$ is 
	satisfiable, there exists a finite model $\eM$ of size
        $s_X(\tau)$ such that $\eM \models \tau$.
\end{definition}
This definition was used in the book \cite{Borger97} in the context of classical logic. 
It can also be used for intuitionistic first order logic, but the finiteness concerns a different but relevant notion of size.
We say that a model $\eM=\tuple{C,\,\leq,\,\{\mathcal A_c: c\in C\}}$ is finite when $C$ is finite and $\mathcal A_c$ is
finite for all $c\in C$. The number $u=|C|+|\bigcup_{c\in C} \mathcal A_c|$ is the size of the model $\eM$.
\begin{lemma}
	For all formulas $\tau$ from a class $X$ that has the finite model property either $\vdash\lnot\tau$  or there exists 
	a finite model $\eM$ and a state $s$ such that $s,\eM\models\tau$.
	\label{lemma:true-or-false}
\end{lemma}
\begin{proof}
	If $\lnot\tau$ or $\tau$ is true the proof follows by the completeness theorem and by definition.
	Otherwise there exists a model $\eM$ of class $X$ such that $s,\eM\not\models\tau$, and a state $s'>s$ such that
	$s',\eM\models\tau$ (if $M$, $s$ and $s'$ do not exist either
        $\tau$ or $\lnot\tau$ would be valid in the model). Note that it does not necessarily mean
	that a proof for $\lnot\tau$ exists.
	But, since Kripke models are monotonous, $(s',\eM)\in X$ and we have a model of $\tau$ in the class $X$: the 
	part of the original model starting in $s'$.\qed
\end{proof}
\begin{definition}[effective class]
  We say that a model $\eM$ is a model of a class of formulas $X$ when
  for each $\varphi\in X$ it holds that $\eM\models \varphi$.
  
  A class of formulas $X$ is \emph{effective} iff, it is decidable
  that given a final model $\eM$ whether $\eM$ is a model of the class
  $X$. 
\end{definition}
For example every class that has a finite number of axioms or axiom schemes, as well as prefix classes from the book of
Börger et al \cite{Borger97} is effective.
\begin{proposition}
	An effective class $X$ has \efmp iff it has \esmp.
\end{proposition}
\begin{proof}
  The implication from right to left is trivial. For the other one: we show how to compute $s_X(\tau)$.
  Given $\tau$ we run two processess in parallel: one generates finite models and checks whether they are models of X
  and then if $\tau$ is satisfied in them, and another one generates proof and checks if one of them proves $\lnot\tau$.
  If the first process succeeds, we return the size of the model found, and if the second suceeds we return 1.
  Lemma~\ref{lemma:true-or-false} shows that one of the processess succeeds. 
  This function is correct as it returns
  the size of a finite model if it exists, and otherwise the formula is not satisfiable, so the return value
  does not matter. \qed
\end{proof}
In the proof above we can also make the function return the smallest finite model by
enumerating the models ordered by size, but it is not needed in this paper.

%

\section{Small model size and small Afrodite strategies}
In this section we show that from a finite countermodel of a given size we can construct a small
Afrodite strategy. First we introduce games, strategies and introduce tools to replace some
variables in a strategy. Then we use these tools to show a limit on the set of eigenvariables in the
game depending on the size of the small model, proving the Theorem \ref{thm:final1}.

\label{sec:small-afrodite-s}

\subsection{Better variables}

\paragraph*{Notation.} We define \emph{substitutions} applied to a formula: $\varphi[x/y]$
is the formula $\varphi$ with all free occurrences of $x$ replaced by $y$ in a capture-avoiding fashion.
The \emph{disjuncts} of a formula $\alpha\lor\beta$ are $\alpha$ and $\beta$, and for a
formula that is not a disjunction the whole formula is called a disjunct. We understand the formula
$\alpha\lor\beta\lor\gamma$ to mean $\alpha\lor(\beta\lor\gamma)$, but understanding it as a disjunction of three
disjuncts would also be possible with minor technical changes.

\paragraph*{Games.} We show how the small model property can be
expressed in terms associated with the notion of intuitionistic games for first-order logic as defined in Section 5 of
the work by Sørensen and Urzczyn \cite{Urzyczyn2016}. 
The game describes a search for a proof and has two players: Eros, tryig to prove the judgement and
Afrodite, showing that it can't be proven.
We write $\Gamma\vdash\tau \gameturnMove \Gamma'\vdash\tau'$
to state that the positions $\Gamma\vdash\tau$ and
$\Gamma'\vdash\tau'$ are connected with a turn.
A \emph{game} is a sequence of \emph{positions} connected by
\emph{turns}, i.e.\ a sequence ${\cal P}_1,\ldots,{\cal P}_n,\ldots$
such that ${\cal P}_i\gameturnMove {\cal P}_{i+1}$ for each $i\in\mathbb{N}$.
 Possible moves are
shown in \figref{fig:moves}.
We omit the subscript ``move'' when it is not needed or clear from the context. 
A game starts in a position $\Gamma\vdash\tau$ and begins with Eros' move, followed by Afrodite's move which determines the next turn. If Eros reaches a \emph{final position} he wins, otherwise the game is infinite and Afrodite wins.
We call $\Gamma\vdash\tau$ the \emph{precedent} and $\Gamma'\vdash\tau'$ the \emph{antecedent} of the move. 
Some turns have players associated with them:
if Afrodite makes a choice in a move we call the precedent an \emph{Afrodite's position}, and if Eros makes a choice we call it an \emph{Eros position}.
A disjunct of $\tau$ is an \emph{aim}. In order to avoid confusion with classical provability we write $\Gamma\vdashFO\tau$ to denote that $\tau$ is
provable from $\Gamma$ in first-order intuitionistic logic. If the exact proof $p$ is important we use the notation
$\Gamma\vdashFOx p\tau$.

\paragraph*{Strategies.} A strategy is a tree of nodes labeled by
positions linked by edges labeled by turns, which we call \emph{moves}. In each position either one or none of the
players makes a choice. 
In a position with no choice the next position is determined by game rules (see Figure~\ref{fig:moves}) and the corresponding
turn must appear in the strategy.
For Afrodite strategy the tree consists of non-final positions and all paths are infinite as well as the tree has at least one
move in each Afrodite's position and all the possible moves (up to renaming of fresh and bound variables) in Eros
positions.
A final position is a position in which $\tau\in\Gamma$ or $\bot\in\Gamma$. 
For Eros strategy all paths end at a final position and the tree has at least one
move in each Eros' position and all the possible moves (up to renaming of fresh and bound variables) in Afrodite
positions.
It should be obvious that if Eros cannot make a move that introduces something new to the game, he is forced to replay one of the previous moves and Afrodite wins.

\paragraph*{Ordering of variables.}
Intuitively speaking we would like to capture the fact that one variable is ``better'' than the other if 
all the information that was known about the ``worse'' variable is kept and possibly extended with new facts.
More formally we say $x_1\preceq_\Gamma x_2$ when for every formula $\tau$
\begin{quote}
	if $\Gamma\vdashFO\tau$ then
	$\Gamma\vdashFO(\tau[x_1/x_2])$.
\end{quote}
The relations $\prec_\Gamma$ and $\sim_\Gamma$ are defined in the following way:
  \[ x_1\sim_\Gamma x_2  \defwhen  x_1\preceq_\Gamma x_2 \land  x_2\preceq_\Gamma x_1, \] %
  \[ x_1\prec_\Gamma x_2 \defwhen  x_1\preceq_\Gamma x_2 \land  x_2\not\sim_\Gamma x_1. \] %
In cases when $\Gamma$ is clear we omit it for brevity.
\begin{proposition}
	For any $\Gamma$, the relation $\preceq_\Gamma$ is a quasiorder, but not a partial order.
\end{proposition}
\begin{proof}
	The relation $\preceq_\Gamma$ is trivially reflexive and transitivity follows immediately from definition
	with help of an observation that \linebreak$\tau[x1/x2][x2/x3]=\tau[x1/x3]$, so it is a quasiorder.

	If we choose two distinct fresh variables $x_\alpha$ and $x_\beta$, i.e.~not in $\FV\Gamma$, we have
	$x_\alpha\preceq x_\beta$ and $x_\beta\preceq x_\alpha$, but $x_\alpha\neq x_\beta$, so $\preceq$ is not
	a partial order.\qed
\end{proof}
This leads to the conclusion that the only important variables are those that are 
\emph{maximal in the $\preceq$ relation}, as we can replace all the other variables with their
maximal counterparts.
\begin{proposition}
	\label{prop:replace_ifol}
	Let $\tau$ be a formula such that $\FV\tau=x_1,\,\ldots,x_n$. If $\Gamma\vdashFO\tau$, then 
	$\Gamma\vdashFO\tau\left[ x_1/x_1',\,\ldots,\,x_n/x_n' \right]$, 
	where, for all $i$, $x_i\preceq x_i'$.
\end{proposition}
\begin{proof}
	We apply the definition of $\preceq$ for each $x_i$ in turn.\qed
\end{proof}

\noindent

\begin{proposition}
	If Eros or Afrodite has a strategy in position $\position=\Gamma\vdash\tau$ 
	and if at some position $\position_{i_0}$ and all subsequent positions in that strategy
	we have $x_1'\preceq x_1$, \ldots, $x_n'\preceq x_n$
	then we can replace all occurrences of variables $x_i'$ with $x_i$ at $\position_{i_0}$ and 
	the same player still has strategy in position $\position_{i_0}$.
	\label{prop:replace}
\end{proposition}
\begin{proof}
	The replacement is done while looking at the whole game tree and with maximum knowledge (i.e.~trueness
	of predicates through the whole game tree) about variables, which is not a problem
	since our aim is to construct the strategy for Afrodite, which implies knowledge of all the possible
	turns.

	Suppose Afrodite has a strategy in $\position$.
	Let us focus on a path in the tree of the strategy
	\[ 
		\Gamma_1\vdash\tau_1 \gameturn \Gamma_2\vdash\tau_2 \gameturn \ldots \gameturn
		\Gamma_n\vdash\tau_n\gameturn\ldots 
	\]
	Each of the moves may add something to $\Gamma$, but nothing is removed and we can separate
	the newly added facts:
	\[ 
			\Gamma_2 = \Gamma_1,\,\psi_1 \quad \dots \quad \Gamma_{n+1} = \Gamma_n,\,\psi_n\\
	\]
	Another view of the new facts would be to separately keep track of those referring to the
	variables $x_i$:
	\[ \Gamma_n = \Gamma_1,\,\Delta_n,\hat\Delta_n \]
	where $\Delta_n$ has all the facts that reference the variables $x_i$ and
	$\hat\Delta_n$ the others. To make the notation concise we write $\Gamma,\Delta_1,\Delta_2\vdash\tau$ 
	as a shorthand for $\Gamma,(\Delta_1\cup\Delta_2)\vdash\tau$.

	Instead of taking the original path we can take the following one
	\[ 
		\Gamma_1\vdash\tau_1 \gameturn \Gamma_2'\vdash\tau_2' \gameturn \ldots \gameturn
		\Gamma_n'\vdash\tau_n'\gameturn\ldots 
	\]
	where $\Gamma_n'=\Gamma_{n-1},\,\psi_{n-1}'$, $\psi_n'=\psi[x_i'/x_i]$ and
	$\tau_n'=\tau_n[x_i'/x_i]$. Or, viewed in the terms of $\Delta$s, 
	\[ \Gamma_n'=\Gamma_1, \Delta_n', \hat\Delta_n \]
	where $\Delta_n'=\Delta_n[x_i'/x_i]$.\\
	To make this construction sound we need to show that
	$\Gamma_1,\,\Delta_n',\,\Delta_n\vdash x_i' \preceq x_i$
	and that the move 
	$(\Gamma_1,\,\Delta_n',\,\hat\Delta_n\vdash \tau_n')
	\gameturn
	(\Gamma_1,\,\Delta_{n+1}',\,\hat\Delta_{n+1}\vdash \tau_{n+1}')$
	is possible.
	The first part follows directly from Proposition~\ref{prop:replace_ifol}.
	For the second part we show how to adapt the original move $\Gamma_n\vdash\tau_n$.
	The possible moves are listed in
	\figref{fig:moves}. Only two of them have direct interaction with non-fresh variables:
	in (a4) and (b5) Eros is free to choose any variable and the replacement variables $x_1, \dots, x_n$
	are already available, so would
	not lead him to winning the game, otherwise he could have played this move in the original
	strategy. 

	The other case is when Eros has a strategy in $\position$.
	The substitution is almost the same as in the previous case except some nonfinal positions might become final,
	as the set of facts known about $x_i'$ is bigger or equal to those that were known about $x_i$, as $x_i'\preceq
	x_i$.
	Final positions remain final by Proposition \ref{prop:replace_ifol}.
	Nonfinal positions might become final, but it only makes Eros win faster.\qed
\end{proof}

\subsection{Construction of the strategy}
\paragraph*{Small strategies of Afrodite.} 
With the aim of relating the size of the Afrodite strategy and the size of the small model we define
a notion of a small strategy. 
Proposition \ref{prop:replace} suggest the following definition.
Since we know that using only the \emph{maximal} variables
is sufficient in the game, we define \emph{small strategy of Afrodite} for a formula $\tau$ from 
class $X$ as a strategy that has at most $s_X(\tau)$ $\simeq$-classes of abstraction of
maximal variables.
Given a small countermodel $\eM$ of a formula we aim to construct a small winning Afrodite strategy $S$, i.e.
one that gives at least one possible response for each possible Eros' move.
For a given turn $\te = \Gamma\vdash\tau$ we need to choose 
a response to Eros moves. We associate a state $s\in\eM$ with each turn $\te$.
Our strategy has the following invariant that holds at each turn:
\begin{equation}
	\exists_{\rho:\FV\Gamma\to \mathcal A_s}\left(  \rho,s\models\Gamma \right) \quad\land\quad \forall_{\rho:\FV\Gamma\to
	\mathcal A_s}\left(  \rho,s\models\Gamma\to \rho,s\not\models\tau\right),
	\label{eq:inv}
\end{equation}
and the sets of maximal variables corresponds to states of the small countermodel. The part of the invariant quantified
with $\exists$ is called the existential part and the part quantified with $\forall$ is called the universal part.

\begin{figure}[htb]
	\noindent Moves manipulating assumptions:
	\begin{itemize}
		\item[*a1)] If $\alpha$ is an assumption $\beta\lthen\gamma$ then Afrodite chooses between positions $\Gamma,\gamma\vdash\tau$ and $\Gamma\vdash \beta$.
		\item[*a2)] If $\alpha$ is an assumption $\beta\lor\gamma$ then Afrodite chooses between positions $\Gamma,\beta\vdash\tau$ and $\Gamma,\gamma\vdash \tau$. 
		\item[a3)] If $\alpha$ is an assumption $\beta\land\gamma$ then the next position is $\Gamma, \beta, \gamma\vdash\tau$.
		\item[a4)] If $\alpha$ is an assumption $\forall x \varphi$ then Eros chooses a variable $y$ and the next position is $\Gamma,\varphi[y/x]\vdash\tau$.
	\end{itemize}
	Moves manipulating the proof goal:
	\begin{itemize}

		\item[a5)] If $\alpha$ is an assumption $\exists x \varphi$ then the next position is $\Gamma,\varphi[y/x]\vdash\tau$ where $y$ is a fresh variable.
		\item[b1)] If $\alpha$ is an aim of the form $\beta\lthen\gamma$ the next position is $\Gamma,\beta\vdash\gamma$.
		\item[*b2)] If $\alpha$ is an aim of the form $\beta\land\gamma$ then Afrodite chooses between positions $\Gamma\vdash\beta$ and $\Gamma\vdash \gamma$. 
		\item[b3)] If the aim $\alpha$ is an atom or a disjunction the next position is $\Gamma\vdash\alpha$.
		\item[b4)] If $\alpha$ is an aim of the form $\forall x \varphi$ the next position is $\Gamma\vdash\varphi[y/x]$ where $y$ is fresh.
		\item[b5)] If $\alpha$ is an aim of the form $\exists x \varphi$ then Eros chooses a variable $y$ and the next position is $\Gamma\vdash\varphi[y/x]$.
	\end{itemize}
	\caption{Table of moves in position $\Gamma\vdash\tau$ for the intuitionistic game \cite[fig.~11, p.~32]{Urzyczyn2016}. In each move Eros
		chooses a formula $\alpha$ - either an assumption or an aim, and the move is selected from this table according to the $\alpha$ 
		chosen.}
	\label{fig:moves}
\end{figure}

\figref{fig:moves} lists possible moves and the choices players make. We define a strategy for Afrodite and
she makes a choice in cases marked with * in the figure. Afrodite should choose in the indicated moves:
\begin{description}
	\item[\textnormal{(a1)}] We choose $\Gamma,\gamma\vdash\tau$ when $\rho,s\models\Gamma,\,\gamma$.
	\item[\textnormal{(a2)}] We choose $\beta$ when $\rho,s\models\Gamma,\,\beta$.
	\item[\textnormal{(b2)}] We choose $\beta$ when $\rho,s\not\models\beta$.
\end{description}
In case of (b1) and (b4) the current model state $s$ might needs to be advanced to some subsequent
state to keep the invariant.

We still need to show that each move preserves the invariant.
\begin{proposition}
	At each position $\position:\Gamma\vdash\tau$  the invariant (\ref{eq:inv}) holds.
\end{proposition}
\begin{proof}
	We assume the notation of Figure~\ref{fig:moves} and show that each move preserves the invariant (\ref{eq:inv}).

	\begin{itemize}
		\item[(a1)] We have two possibilities:
			\begin{itemize}
				\item In case $\rho,s\models\Gamma,\gamma$: we choose $\Gamma,\gamma\vdash\tau$. The existential
					part of the invariant follows directly from the invariant of the previous step. For the 
					universal part suppose the opposite, i.e $\rho,s\models\tau$, so for given $\rho$ we
					either have contradiction with $\rho,s\not\models\tau$ from invariant of the previous
					step or $\rho,s\not\models\gamma$, but then we would not choose this move for the
					strategy.
				\item Otherwise $\rho,s\not\models\Gamma,\gamma$ and we choose $\Gamma\vdash\beta$. The
					existential part of the invariant remains true as $\Gamma$ does not change. For universal
					part suppose $\rho,s\models\beta$, but then $\rho,s\models\gamma\to\beta$, which
					is in contradiction with the invariant from the previous step.
			\end{itemize}
		\item[(a2)] Once again we have two possibilities:
			\begin{itemize}
				\item In case $\rho,s\models\Gamma,\beta$: we choose $\Gamma,\beta\models\tau$. The existential
					part of the invariant follows directly from the invariant of the previous step. The 
					universal part is the same as in the corresponding point of the move (a1).
				\item Otherwise $\rho,s\not\models\Gamma,\gamma$ and the proof is the same as in the
					corresponding point of (a1).
			\end{itemize}
		\item[(a3)] The existential part is true because $\rho,s\models\beta,\gamma$ follows from $\rho,s\models\beta\land\gamma$.
			The universal part is proven by simply applying the definition of $\models$.
		\item[(a4)] We can choose any value for $\rho(y)$.
			Existential part: from the invariant in the previous move we have $\rho,s\models\forall x\varphi$, 
			so we apply definition of $\models$ to get $\rho,s\models\varphi[y/x]$.
			Universal part: suppose that $\rho,s\models\Gamma,\varphi[y/x]$ and $\rho,s\models\tau$. But this means
			$\rho,s\models\Gamma$ which implies have a contradiction with $\rho,s\not\models\tau$ from the previous move.
		\item[(a5)] Since $\rho,s\models\exists x\varphi$ we know that there exists $\hat x$ such that $\rho[\hat x/x],s\models\varphi$.
			In the existential part we just need to take $\rho(y)=\rho(\hat x)$. 
			Universal part: identical with the universal part of (a4).
		\item[(b1)] Using the assumption we have a state $s'\geq s$ such that $\rho,s'\models\Gamma,\beta$ but
			$\rho,s'\not\models\gamma$. We advance $s$ to $s'$. The existential part is trivially true.
			For the universal part: suppose $\rho,s'\models\Gamma,\beta$ and $\rho,s'\models\gamma$. This is 
			in contradiction with $\rho,s\not\models\beta\to\gamma$.
		\item[(b2)] We have the following cases:
			\begin{itemize}
				\item $\rho,s\not\models\beta$: the set of assumptions does not change so the existential part
					is proven by applying the existential part from the previous move. For the universal part,
					$\rho,s\not\models\beta$ is exactly the assumption of the case under investigation.
				\item otherwise  $\rho,s\models\beta$. We choose the position $\Gamma\vdash\gamma$; the 
					existential part is the same as in the previous step. For the universal part suppose 
					$\rho,s\models\gamma$: then $\rho,s\models\beta\land\gamma$ contradicts
					the invariant $\rho,s\not\models\beta\land\gamma$ from the previous move.
			\end{itemize}
		\item[(b3)] The existential part is the same as in the previous move. The universal part is the same as in the 
			second bullet of (b2).
		\item[(b4)] We can choose any value for $\rho'(y)$. 
			The existential part is true since $\Gamma$ is the same as previously and
			the valuation of $y$ does not affect it. 
			The universal part: suppose that $\rho',s\models\varphi[y/x]$. Then by definition $\rho,s\models\varphi$,
			which is in contradiction with the invariant from the previous move.
		\item[(b5)] The existential part is the same as in the previous move. For the universal part suppose $\rho,s\models
			\varphi[y/x]$. Then by definition of $\models$ we have $\rho,s\models\exists x \varphi$, which is 
			a contradiction with the invariant of the previuos step.
	\end{itemize}\qed
\end{proof}

The constructed strategy is \emph{small}: elements of $\mathcal A_s$ correspond to $\simeq$-classes and 
the valuation $\rho$ proves that all the variables fit in $s_X(\varphi)$ classes as the size of the model is $s_X(\varphi)$.
This proves the following:
\begin{theorem}
	For all classes $X$ that have the small model property and all formulas $\tau\in X$, if a strategy of Afrodite exists for $\tau$ then
	a small strategy of Afrodite for $\tau$ also exists. \label{thm:final1}
\end{theorem}

\section{Small model size and the Arcadian automata}
\label{sec:small-automata}
Here we show a limit on resources of Arcadian automata \cite{poprzedni} checking derivability of a formula
$\varphi$ from an effectively axiomatized class $X$ that has the finite model property.
Theorem~\ref{thm:final2} shows a limit on the size of the set of
eigenvariables of the automaton in terms of the numbers of variables and subformulas in $\varphi$
and the size of the small model. We reason only about automata that are
translated from a formula as defined in Section 4 of \cite{poprzedni}. 
In Section~\ref{sec:arcadian-automata} we introduce Arcadian automata, show how to replace 
variables in their runs in Section~\ref{sec:arcadian-replace}, and limit the size of 
the set of eigenvariables in \ref{sec:arcadian-final}.

\subsection{Arcadian automata}
\label{sec:arcadian-automata}
\paragraph*{Notation} We already know that $\simeq$-maximal variables play a crucial role. Given a set of facts $\Gamma$
we denote by $\check\Gamma$ the set obtained from $\Gamma$ by selecting only those facts $\gamma$ that do have only
maximal variables in $\FV\Gamma$. 
An \emph{Arcadian automaton} is a tuple
$\< \mathcal A, Q, q^0, \varphi^0, \mathcal I, i, \mathrm{fv}\>$, where $\mathcal A$ is a finite tree, $Q$ and $I$ are sets of states and 
instructions with $i$ mapping states to instructions, $\mathrm{fv}:A\to P(A)$ describes the binding of variables and $q^0$
and $\varphi^0$ are the inital state and node. The function $\mathrm{fv}$ satisfies the condition that for all
$v$ either $v$ is a leaf or $\mathrm{fv}(v) = \bigcup_{w\in B(v)}\mathrm{fv}(w)$ where $B(v) = \{
w \mid v\succc w\}$ and $\succ$ is the usual predicate of being a successor.
An \emph{instantaneous description} is $\< q, \kappa, V, w, w', S\>$ where
$q\in Q$  and $\kappa\in A$ are the current state and node, $V$ is a set of
eigenvariables, $w$ and $w'$ are interpretations of bindings and $S$ is the
store. For more details see \cite{poprzedni}.

\subsection{Better variables in Arcadian automata}
\label{sec:arcadian-replace}
\paragraph*{Equivalent positions} 
We say that the position $\Gamma\vdash\tau$ and $\Gamma'\vdash\tau'$ are \emph{equivalent} when $\check\Gamma=\check\Gamma'$ and $\tau=\tau'$.

\begin{proposition}
	Suppose $\Gamma,\hat\Gamma\vdashFO M:\tau$ where for some $\alpha$ and $\alpha'\succeq\alpha$ such that $\alpha\in\FV\Gamma$ and
	$\alpha\not\in\FV{\hat\Gamma}$. If $\Gamma=x_1:\tau_1,\dots,\Gamma$ $\vdashFO x_n:\tau_n$ then
	$\hat{\Gamma},\Gamma,\Gamma'\vdashFO M[x_1/x_1']\dots[x_n/x_n']:\tau[\alpha/\alpha']$ where 
	$\Gamma'=x_1':\tau_1[\alpha/\alpha'],\dots,x_n':\tau_n[\alpha/\alpha']$ and $x_1',\dots,x_n'$ are fresh
	variables, i.e. $x_i'\not\in\FV{M}$.
\end{proposition}
\begin{proof}
	Proof is by induction over the length of the proof of $\tau$. We look at the last rule in the proof. 
	In most of the cases the conclusion follows by simple application of the inductive hypothesis,
	but there are three rules that change the environment, namely
	($\lor E$), $(\to I$) and ($\exists E$) and the proof is more subtle for them.
	Let us focus on the $(\to I)$ rule. If $(x_i:\tau_i)\in\hat\Gamma$ we do not need to change anything,
	in the other case we know that $(x_i:\tau_i)\in\Gamma$ and we apply the induction hypothesis and use the assumption
	$(x_i':\tau_i[\alpha/\alpha'])\in\Gamma'$ for the $\lambda$-abstraction.
	We can now remove the variable $x_i:\tau_i$ as it is not referenced in
	$M[x_1/x_1']\dots[x_n/x_n']:\tau[\alpha/\alpha']$.
	
	The induction base is the \emph{var} rule, since the proof must
	begin with this rule, and the correctness of replacing $\alpha$ with $\alpha'$ follows immediately from
	the definition of $\Gamma'$.\qed
\end{proof}
\begin{proposition}
	If $\Gamma\vdashFO M:\tau$, $\alpha$ and $\alpha'$ are variables in $M$ such that
	$\alpha\preceq\alpha'$ then $\Gamma\vdashFO M[\alpha/\alpha']:\tau[\alpha/\alpha']$ and $\alpha\not\in FV(\tau)$.
	\label{prop:proof-replace-one}
\end{proposition}\vspace{-1.6em}
\begin{proposition}
	If $\Gamma\vdashFOx{p} M:\tau$ then there exists $M'$ and $p'$ such that $\Gamma\vdashFOx{p'} M':\tau$ and
	in each step $\Gamma'\vdash\tau'$ of $p'$ only maximal variables are mentioned in $\tau'$.
\end{proposition}\vspace{-1.6em}
\begin{proof}
	Apply Proposition~\ref{prop:proof-replace-one} sequentially to each nonmaximal variable.\qed
\end{proof}\vspace{-2em}

\subsection{Loquacious runs}
\label{sec:arcadian-final}
Note that our logic has the subformula property (\cite{poprzedni}). This means that there is only a limited
number of possible targets $\tau$. Let us review fragments of an automaton run 
$p_0\to_\alpha\Gamma\vdash\tau\to_\beta\Gamma'\vdash\tau\to_\gamma p_1$. If $\Gamma\vdash\tau$ and $\Gamma'\vdash\tau$
are equivalent, then a part of the run $\to_\beta$ is removable, i.e. there exists a run of the same automaton
$p_p\to_\alpha\Gamma\vdash\tau\to_{\bar\gamma}\bar p_1$ where
$\bar \gamma$ and $\bar p_1$ are obtained from $\gamma$ and $p_1$ by replacing some variables by their maximal counterparts.
Otherwise that fragment is not removable, as new fact about the maximal variables are discovered, but the number of 
such non-removable runs is limited: there are at most $s_X(\varphi)$ maximal variables. Suppose $v(\varphi)$ is the number
of variables in $\varphi$ and $f(\varphi)$ is the number of subformulas in $\varphi$. The maximum size of 
an environment $\Gamma$ is $\mu=f(\varphi)\cdot s_X(\varphi)^{v(\varphi)}$. Each non-trivial step has to either add something
to the environment or change the target $\tau$ as otherwise the previous state is repeated.
We have at most $\mu$ possible targets, so after at most $\mu$ steps the target
repeats and in the worst case each step introduces a new variable, so the maximum size of $V$ in the automaton is $\mu^2$.

This proves the following:
\begin{theorem}
	Let $\tau$ be a formula from an effective class $X$ that has the small model property. For a given accepting run
	of an Arcadian automaton for that formula there exists an accepting run of the same automaton with the same
	result that has the property $|V|\leq\mu^2$, where $V$ is the working domain of the automaton and $\mu$ is the maximum
	size of the environment defined in the previous paragraph.
	\label{thm:final2}
\end{theorem}

\section{Conclusion}

The small model size is directly related to important
resources in games and automata for checking provability.  In terms of
games, the elements of models directly correspond to abstraction
classes of maximal elements of a quasiorder on eigenvariables that captures the
relation of having more information available about a variable.

For automata the number of such maximal elements can be directly related to the size of set
of eigenvariables $V$; the dependency is exponential, caused by the necessity of
representing the eigenvariables that correspond to non-maximal elements of the quasiorder.

These observations lead to an idea for implementing proof theory bases proves in a manner
that would not be substantially less powerful than those based on model theory. More specifically we
suggest that $V$ should not be represented syntactically but rather as an abstraction class of the
quasiorder.

\bibliography{finmodel}

\end{document}